\documentclass[11pt,a4wide]{article}
\pdfoutput=1
\usepackage{amssymb,a4wide}
\usepackage{amsmath}
\usepackage{latexsym,url}
\usepackage{ascii}
\usepackage{sgamevar}
\usepackage{color}
\usepackage{xspace}
\usepackage{ntheorem}
\usepackage{dsfont}
\usepackage{extarrows}
\usepackage{authblk}

\usepackage{tikz}
\usetikzlibrary{automata}
\usepackage{pgf,pgfarrows,pgfnodes}
\usepackage{multicol,xspace}
\usepackage {pgflibraryshapes}
\usepackage{eurosym}
\usetikzlibrary{arrows,backgrounds,positioning,fit,calc,shadows}

\newtheorem{theorem}{Theorem}

\newtheorem{algorithm}[theorem]{Algorithm}

\newtheorem{definition}[theorem]{Definition}

\newtheorem{proposition}[theorem]{Proposition}

\newenvironment{proof}[1][Proof]{\textbf{#1.} }{\ \rule{0.5em}{0.5em}}

  \newcommand{\cut}[1]{}
  \newcommand{\vcut}[1]{}

\newcommand{\vnote}[1]{}
\newcommand{\pnote}[1]{}

\newcommand{\IF}{\mbox{   \textsc{if}   }}

\renewcommand{\phi}{\varphi}

\def\red#1{\colorbox{red}{#1}}
\def\blue#1{\colorbox{cyan}{#1}}
\def\purple#1{\colorbox{purple}{#1}}
\def\yellow#1{\colorbox{yellow}{#1}}

\begin{document}

\title{Endogenous games with goals: \\ side-payments among goal-directed artificial agents}


\author{Paolo Turrini}

\affil{Department of Computing, Imperial College London}


\maketitle

\begin{abstract} 
Artificial agents, of the kind studied in AI, are typically oriented to the realization of an externally assigned task and try to optimize over secondary aspects of plan execution such time lapse or power consumption, technically displaying a quasi-dichotomous preference relation. 
Boolean games have been developed as a paradigm for modelling societies of agents with this type of preference.
In boolean games agents exercise control over propositional variables and strive to achieve a goal formula whose realization might require the opponents' cooperation. Recently, a theory of {\em incentive engineering} for such games has been devised, where an external authority steers the outcome of the game towards certain {\em desirable} properties consistent with players' goals, by imposing a taxation mechanism on the players that makes the outcomes that do not comply with those properties less appealing to them. The present contribution stems from a complementary perspective and studies, instead, how games with quasi-dichotomous preferences can be transformed from inside, rather than from outside, by endowing players with the possibility of sacrificing a part of their payoff received at a certain outcome in order to convince other players to play a certain strategy. Concretely we explore the properties of {\em endogenous games with goals}, obtained coupling strategic games with goals, a generalization of boolean games, with the machinery of {\em endogenous games} coming from game theory. We analyze equilibria in those structures, showing the preconditions needed for desirable outcomes to be achieved without external intervention. Finally, making use of taxation mechanism as introduced in the literature, we show how to transform these structures in such a way that desirable outcomes can be realized even when side-payments are allowed. What our results show is that endogenous games with goals display specific irreducible features --- with respect to what already known for endogenous games --- which makes them worth studying in their own sake.

\end{abstract}

\maketitle
\section{Introduction}

A characteristic feature of agents, the kind of artificial entities studied in AI, is that of being goal-oriented \cite{GPS, russell-norvig, cristiano, conte-castelfranchi}, directed to the realization of certain desirable states of affairs, typically externally assigned, disregarding secondary factors, such as amount of resources spent or monetary rewards received, whenever they hinder a goal state to obtain. As a consequence of this, an agent does not automatically respond to incentives, as an economic actor would, but displays a typical quasi-dichotomous preference relation over possible outcomes, i.e., it compares states looking at the realization of their own goal first, and only then at efficiency levels.

A well-known model of a society of goal-directed agents heavily investigated in AI is boolean games  \cite{paul}, a compact and computationally desirable representation of strategic interaction by means of logical formulas. In boolean games agents exercise control over propositional variables and strive to achieve a goal formula whose realization might require the opponents' cooperation, disregarding the cost associated to each action, if need be.

 Recently a theory of {\em incentive engineering} has been devised \cite{mike}, where an external authority, i.e., the principal, steers the outcome of the game towards certain {\em desirable} properties, by imposing a taxation mechanism on the players  \footnote{The term {\em player}, mutuated from game theory, and the term {\em agent}, mutuated from AI, will be used interchangeably.} that makes the outcomes that do not comply with those properties less appealing to them. However, that of a principal turns out to be a non-trivial task as, due to preference quasi-dichotomy, there is no monetary compensation that can convince agents to give up their goal. In all the other cases, though, players behave as cost-minimizers and desirable systemic properties satisfying players' goals have been shown to be implementable by the appropriate system of incentives. 

The present contribution \footnote{This paper generalizes and significantly extends \cite{EBG}.} stems from a complementary perspective and studies, instead, how games with quasi-dichotomous preferences can be transformed from inside, rather than from outside, by endowing players with the possibility of sacrificing a part of their payoff received at a certain outcome in order to convince other players to play a certain strategy. Concretely we explore the properties of {\em endogenous games with goals}, obtained coupling strategic games with goals, a generalization of boolean games \cite{paul, mike}, with the machinery of {\em endogenous games} coming from game theory \cite{JW05}. We analyze equilibria in these novel structures, showing the preconditions needed for desirable outcomes to be achieved without external intervention .  We illustrate our idea and, informally, our setting in the following example.

\paragraph{Motivating example}
Consider two players, $a$ and $b$, which can decide whether two light switches $s_a$ and $s_b$ are {\em on} or {\em off}. Let us assume player $a$ to be in full control of $s_a$, player $b$ of $s_b$, and that each player is unaware of the other player's final decision. Let us also assume that players have goals and actions have costs, in particular that
$a$ wants $s_b$ to be {\em on}, that $b$ wants both $s_a$ and $s_b$ to be {\em on}, and that the cost of turning $s_a$ {\em on} is $5$, of turning it {\em off} is $4$, while the cost of turning $s_b$ either {\em on} or {\em off} is $2$. Finally, we assume that players always prefer to minimize the cost of the actions they take and that they always prefer outcomes satisfying their goal to outcomes that do not.
It goes without saying that in a scenario of this kind player $a$, being indifferent between having $s_a$ {\em on} or {\em off}, will simply look at the resulting costs and decide to turn the switch {\em off}. Likewise, player $b$ will turn  the switch {\em on} and this fact will not impose on him any extra cost. In the end, player $a$ will have her goal satisfied, paying a cost of $4$, while $b$ will not, paying $2$.

Suppose though that, before the game starts, players can commit to bear a part of the cost the opponent is incurring in at a certain outcome, should that outcome be reached.
While in the previous scenario player $b$, although {\em depending} on the other player for the realization of his goal (cfr. \cite{elise} for a qualitative account of dependencies in boolean games and \cite{grossi-turrini-dep} for strategic games), could not have any say on $a$'s decision-making, now he can  adopt a richer strategy and offer $a$ to bear some cost of $s_a$ being {\em on}, say 3, before $a$ takes any decision. In the resulting situation both $a$ and $b$ will have their goal satisfied, $a$ bearing a cost of $2$ while $b$ of $5$, which is a more satisfactory solution for both players. Notice, however, that the solution is not stable, as player $b$ has an incentive to deviate to more parsimonious offers in the pre-play phase without compromising the realization of his own goal.

\medskip

The added value of the analysis presented here, intuitively introduced in the example above, is two-fold:
\begin{itemize}

\item it complements the framework of incentive engineering for boolean games \cite{mike}, studying those situations in which players can reach desirable properties {\em without} external intervention;
\item it provides a quantitative resolution to dependence relations, broadly studied for the case of boolean games \cite{elise}\cite{elise2}, allowing players to influence each other's decision-making by the offer of monetary incentive. 

\end{itemize}

We carry out the analysis studying the general setting of endogenous games with goals, in relation with both boolean games and endogenous games, focussing on the properties of the resulting equilibria.

\paragraph{Paper Structure} 

In Section \ref{sec:quasi} we introduce strategic games with goals, studying their formal connection with related contributions in the game theory and AI literature, i.e., strategic (normal form) games and boolean games. In Section \ref{sec:endo} we study endogenous games with goals, adding to strategic games with goals the dynamics brought into play by the possibility of exchanging side-payments in the pre-play phase. In Section \ref{sec:equilibrium} we carry out an equilibrium analysis of these structures, showing results of pure strategy equilibrium survival and discussing the connection with what known from game theory. In Section \ref{sec:integration} we integrate side-payments with taxation mechanisms, devising a procedure that ensures desirable properties to be reached. Section \ref{sec:lexico} studies a variant of quasi-dichotomous preferences with a lexicographic order over mixed strategies, discussing the connection with the standard utility function of normal form games. Finally, in Section \ref{sec:conclusion} we wrap up the work pointing to possible future research directions.

\section{Strategic games and quasi-dichotomous preferences}\label{sec:quasi}

In this section we describe a general approach to characterizing goal-directed artificial agents acting in a common shared world. We do so by explicitly enriching a strategic game with a distinguished set of goals, one for each agent. 

As well-known, a {strategic (normal form) game} $\mathcal{S}$ is a tuple $(N, \{\Sigma_i\}_{i \in N}, \pi)$, where $N$ is a set of players, $\Sigma_i$ a set of strategies for player $i$ and $\pi: \prod_{i \in N} \Sigma_i \times N \to \mathbb{R}$ a {payoff function}, assigning to each player his payoff at each strategy profile. Henceforth we abbreviate $\pi(\sigma,i)$ as $\pi_i(\sigma)$, $\prod_{i \in N} \Sigma_i$ as $\Sigma$, --- extending the conventions to similar cases ---and denote $NE(\mathcal{S})$ the set of pure strategy Nash equilibria of strategic game $\mathcal{S}$, with $NE^{\Delta}(\mathcal{S})$ being its mixed extension.

Strategic games with goals are defined as follows.

\begin{definition}[Strategic games with goals]

A {\em strategic game with goals} is a tuple $(\mathcal{S},\{G_i\}_{i\in N})$ where $\mathcal{S}=(N, \{\Sigma_i\}_{i \in N}, \pi)$ is a strategic game and each $G_i \subseteq \Sigma$ is a set of goal states for player $i$.

\end{definition}

Intuitively a {\em goal} is a state of the game the agent is directed to and, when given the possibility, would not want to trade for a non-goal state, no matter what the payoff assigned by the function $\pi$ is.

If we think of each agent as associated to a colour, we obtain a particularly intuitive representation of strategic games with goals, where goal states for each agent are assigned the agent's colour. The algebra of goal states is defined by operation on colours, e.g. a goal state shared by a blue player and red player is represented as a purple state. Figure \ref{PD} is an example of how strategic games with goals can be displayed.

\begin{figure}[htb]\hspace*{\fill}%
\begin{game}{2}{2}
     \> $L$   \> $R$\\
$U$ \> {${3,3}$}\>{${0,5}$}\\
$D$   \> \blue{$5,0$}   \>\purple{$1,1$}
\end{game}\hspace*{\fill}%
\caption{Players' goals. \textcolor{cyan}{Column} wants the game to end up in the set $\{(L,D),(R,D)\}$, \textcolor{red}{Row} in the set $\{(R,D)\}$. The coalition $\{\textcolor{purple}{Column}, \textcolor{purple}{Row}\}$ wants the shared outcome $(R,D)$ to be realized. As a convention, \textcolor{red}{Row} obtains the first component in the displayed payoff vectors, \textcolor{cyan}{Column} the second.}
\label{PD}
\end{figure}

\subsection{Quasi-dichotomous preferences}

We stress it once more: goal states represent those outcomes that are of utter importance for an agent. In particular, when confronted between the choice of a goal state and a non-goal state, the agent will always elect to choose the goal state, even if this amounts to giving up secondary rewards, expressed by the payoff function. When goal realization is however not an issue, secondary aspects play a role and the agent will always try to maximize the resulting payoff.

This fact induces what is technically called {\em quasi-dichotomy} of a preference relation \cite{mike}: there is a distinguished set of states that is better for an agent than all the others (the goal states) {\em and will remain so}, independently of money transfers. Both goal states and non-goal states can in any case be ordered looking at the payoff that each player is associated to.

Clearly, preference quasi-dichotomy could be implemented in many ways. The first that comes to mind is that of working with truly lexicographic preferences \cite{Rubinstein}, i.e., states being identified with a tuple $(x,n)$ where $x\in \{0,1\}$ and $n \in \mathbb{R}$, the first entry encoding whether the state is a goal state or not and the second entry encoding secondary materialistic aspects. Therefore, a state $z$ is to be preferred to a state $z^{\prime}$ whenever $z \geq^{LEX} z^{\prime}$, where $\geq^{LEX}$ is the lexicographic order between the two.
Under this interpretation, goal states are {\em de facto} assigned an infinite payoff, which makes them better than a non-goal state no matter what the payoff of the latter is.

It is well-known that both games with lexicographic preferences and games with infinite utility do not have a corresponding von Neumann-Morgenstern utility representation \cite{Wilson, Vicky, Brandeburger, Rubinstein} and the resulting games do not in general display fundamental game-theoretical properties, such as existence of Nash equilibria. Furthermore, games with infinite utility rule out the possibility of having standard expected utility to begin with and games with lexicographic preferences allow only for non-standard versions thereof. Section \ref{sec:lexico} will carry out an in-depth analysis of these non-standard ways of dealing with quasi-dichotomy, showing how even natural representations of mixed strategies cannot be analyzed within the framework of normal form games, preventing therefore a comparison with well-known setups coming from game theory, e.g. \cite{JW05}.

To overcome these problems we study a generalization of boolean games with costs (as introduced in \cite{mike} and analyzed further in \cite{EBG}) which are instead game-representable. Concretely, we define a family of {\em boost factors} $\Omega_i$, one for each player $i$, each of which encodes how much more a player values a certain goal state with respect to a non-goal state. In other words, a boost factor is a measure of the relative distance that, at each game, a certain goal state for an agent finds itself with respect to all other non-goal states.
Notice that boost factors are mechanisms to ensure that goal states {\em remain} better than non-goal states, but they implicitly also give a measure of the {\em risk} that a player is willing to undertake to achieve a goal state, i.e., they encode a preference relations over mixed profiles containing both goal states and non-goal states.

Technically, for a given strategic game with goals  $(\mathcal{S},\{G_i\}_{i\in N})$, a {\bf boost factor} is a function $\omega^{(\mathcal{S},\{G_i\}_{i\in N})}_i: \mathbb{R} \to \mathbb{R}$ associating to each payoff how much this payoff is {\em boosted} if it is the payoff of a goal state. Properties required by boost factors are the following.

 Let $x,y \in \mathbb{R}$ and $(\mathcal{S},\{G_i\}_{i\in N})$ be a strategic game with goals.
 
 Each {$\omega^{(\mathcal{S},\{G_i\}_{i\in N})}_i: \mathbb{R} \to \mathbb{R}$} is required to be such that:

\begin{equation}
\omega^{(\mathcal{S},\{G_i\}_{i\in N})}_i(x) \geq \omega^{(\mathcal{S},\{G_i\}_{i\in N})}_i(y) \mbox{ if and only if } x \geq y
\end{equation}
Intuitively, if a state $x$ is weakly preferred by $i$ to a state $y$ for its secondary aspect then it remains so whenever both $x$ and $y$ satisfy $i$'s goal.

\begin{equation}
\omega^{(\mathcal{S},\{G_i\}_{i\in N})}_i(x) > \pi_i(\sigma) \mbox{ for all } \sigma \not\in G_i
\end{equation}

Intuitively, satisfying a goal is always better than not satisfying it.

A strategic game with goals $(\mathcal{S},\{G_i\}_{i\in N})$ associated to profile of boost factors $\omega \in \Omega = \prod_{i\in N} \Omega_i$ is said to be {\bf instantiated by} $\omega$, and this is denoted $(\mathcal{S},\{G_i\}_{i\in N})(\omega)$. Intuitively, when playing $(\mathcal{S},\{G_i\}_{i\in N})(\omega)$ each agent $i$ judges the betterness of goal states according to the boost factor $\omega_i$. Here is the definition of utility taking them into account.

\begin{definition}[Utility]\label{def:utilities-boost} Let $(\mathcal{S},\{G_i\}_{i\in N})(\omega)$ be a strategic game with goals instantiated by a profile of boost factors $\omega$ with $\mathcal{S} = (N, \Sigma, \pi)$. The {\bf utility function} $u^{(\mathcal{S},\{G_i\}_{i\in N})(\omega)}: N \times \Sigma \to \mathbb{R}$ assigning to each player the payoff $u^{(\mathcal{S},\{G_i\}_{i\in N})(\omega)}_i(\sigma)$ he receives at outcome $\sigma$ is defined as follows.

$
u^{(\mathcal{S},\{G_i\}_{i\in N})(\omega)}_i(\sigma)= \left\{
\begin{array}{ll}
\omega(\pi_i(\sigma))& \IF \sigma \in G_i\\
\pi_i(\sigma) & \mbox{otherwise}
\end{array}
\right.
$
\end{definition}

So, the function $u$ is constructed by the combination of $\pi$ and $\omega$, i.e., the payoff function and the profile of boost factor, respectively. The latter takes care of the fact that goal states are always better than non-goal states, no matter what the payoff is associated to the latter by the function $\pi$. 

We would like at this point to clarify the possible conceptual ambiguity that might arise from having two different functions, $\pi$ and $u$, which associate a vector of numerical values to each outcome. The function $\pi$, which we will always refer to as a {\em payoff function}, encodes the secondary, intuitively purely monetary, aspects of a certain state. The function $u$ instead, which we will always refer to as a {\em utility function}, incorporates the primary aspects, i.e., goal realization, possibly associated to a state. Thereby, a state might have (relatively) high utility and (relatively) low payoff, if for instance the state satisfies a goal, but it might also have (relatively) low utility and (relatively) high payoff, if for instance it is the only state not satisfying a goal.

It is also worth noticing that boost factors allow us to reason about hypothetical utility distributions (all the real numbers that are not occurring as payoffs in the game). This extremely important feature will be fully exploited later on, as the role of boost factors is not only to declare that goal states are better, but, again, to keep them so independently of any monetary compensation.

For a given strategic game with goals $(\mathcal{S},\{G_i\}_{i\in N})$, with $\mathcal{S} = (N,\Sigma,\pi)$, and instantiated with boost factor profile $\omega$, the {\bf induced strategic game} is the game $\mathcal{S^{\prime}} = (N,\Sigma,u)$, where $u$ is calculated according to Definition \ref{def:utilities-boost}.

Figure \ref{trans} shows how.

\begin{figure}[htb]\hspace*{\fill}%
\begin{game}{2}{2}
    \> $L$  \> $R$\\
$U$ \> {${-3,-3}$}\>$0,-5$\\
$D$   \> \blue{$-5,0$}   \>\purple{$-1,-1$}
\end{game}\hspace*{\fill}%
\hspace{2mm} 
$\pi_i(\sigma) \Rightarrow{u_i(\sigma)}$
\hspace{2mm} 
\begin{game}{2}{2}
     \> $L$   \> ${R}$\\
$U$ \> {${-3,-3}$}\>$0,-5$\\
$D$   \> {$-5,1$}   \>{$3,0$}
\end{game}\hspace*{\fill}%
\caption{From a strategic game with goals to its induced strategic game: each agent's boost factor assigns $+3$ to all his goal states with respect to his best non-goal state, maintaining the relative distance among goal states. For example, the reason why \textcolor{cyan}{Column} is getting $1$ at outcome $(D,L)$ of the induced strategic game is because he is getting $0$ at outcome $(D,R)$ --- which is in turn because $-3$ is the payoff of his best non-goal state and $-3+3=0$ --- and the original relative distance between $(D,R)$ and $(D,L)$ is of $1$.}
\label{trans}
\end{figure}

\subsubsection{Expected Utility}

We introduce an extra novel feature with respect to the standard treatment of boolean games: we allow players to randomize over possible strategies. This will make it possible to draw a comparison with the equilibrium existence results known for endogenous games, which rely on Nash's theorem, the well-known result on the existence of Nash equilibria with mixed strategies in normal form games \cite{OR, nash}. 
To compute expected utility we first denote $\Delta(\Sigma_i)$ the set of probability distributions over the strategies of player $i$ and, for $\delta \in \Delta(\Sigma_i)$, we denote $\delta(\sigma_i)$ the probability that  {\bf mixed strategy profile}  $\delta$ assigns to $\sigma_i \in \Sigma$. We call the set $s(\delta)=\{\sigma_i \in \Sigma_i \mid \delta(\sigma_i) > 0\}$ the {\bf support} of $\delta$.

\begin{definition}[Expected Utility]

Let $\delta$ be a mixed strategy profile available at strategic game $\mathcal{S}$ with a payoff function $\pi$, $\sigma \in s(\delta)$ a pure strategy profile in the support of $\delta$, and $\delta(\sigma)$ the probability of $\sigma$ to occur according to $\delta$. The {\bf expected utility} of $\delta$ for player $i$ is defined as follows.

$$E_i(\delta)=\sum_{\sigma\in s(\delta)} \pi_i(\sigma) \delta(\sigma)$$

\end{definition}

To compute the expected utility on a strategic game with goals, instantiated with a boost factor, we compute the expected utility in its induced strategic game.

\subsubsection{Boolean games}

An instance of strategic games with goals are boolean games.

\begin{definition}[Boolean Games] A {\bf boolean game }is a tuple  $$(N, \Phi, c, \{\gamma_i\}_{i\in N}, \{\Phi_i\}_{i\in N})$$ where: 

\begin{itemize}

\item $N$ is a finite set of players; 
\item $\Phi$ a finite set of propositional atoms; 
$c:N \times V \rightarrow \mathbb{R}_+$ is a {\bf cost function}, associating to each player the cost he incurs in when some valuation $v \in V$ of the atoms $\Phi$ obtains;

\item $\gamma_i$ is a boolean formula, constructed on the set $\Phi$, denoting the {\bf goal} of player $i$;
\item $\Phi_i \subseteq \Phi$ is the nonempty set of atoms controlled by player $i$. As standard \cite{wiebePC}, we assume that for $j \neq i$, $\Phi_i \cap \Phi_j = \emptyset$ and that $\bigcup \{\Phi_i \mid i \in N\} = \Phi$, i.e.,  controlled atoms partition the whole space.

\end{itemize}

\end{definition}

\medskip
A {\bf choice} of player $i$ is a function $v_i: \Phi_i \to \{{\bf tt}, {\bf ff}\}$, representing player $i$'s decision to set the atoms he controls to either true or false . We denote $V_i$ as the set of possible choices of player $i$. An {\bf outcome} $v \in \prod_{i \in N} V_i$ of the boolean game ${\mathcal{B}}$ is a collection of choices, one per player. An outcome induces a valuation function, assigning a boolean value to each propositional atom \cite{wiebePC}. Therefore, we denote $V=\prod_{i \in N} V_i$ the set of all possible valuation functions and, for $v\in V$ and $\phi$ being a boolean formula constructed on $\Phi$, we write $v \models \phi$ (resp. v $\not\models \phi$) to say that $\phi$ holds (resp. does not hold) under the valuation $v$. When a formula $\phi$ is satisfied by a unique valuation, we use $v_\phi$ to denote that valuation.

The utility in boolean games is calculated using a {boost factor} $\mu_i$, with $\mu_i= max_{v \in V} ( c_i(v))$, selecting the payoff of the worst outcome that can happen to player $i$, i.e., the updated valuation that is most costly to him, and adding it to the goal states, together with a sufficiently small real $\epsilon$.

$
u^{\beta(\mathcal{B})}_i(v)= \left\{
\begin{array}{ll}
\epsilon + \mu_i  - c_i(v)& \IF v \models \gamma_i\\
- c_i(v) & \mbox{otherwise}
\end{array}
\right.
$

Figure \ref{trans-bool} shows an example of this translation.

\begin{figure}[htb]\hspace*{\fill}%
\begin{game}{2}{2}
     \> $s_{C}$   \> $\neg s_{C}$\\
$s_{R}$ \> {${3,3}$}\>$0,5$\\
$\neg s_{R}$   \> \blue{$5,0$}   \>\purple{$1,1$}
\end{game}\hspace*{\fill}%
\hspace{2mm} 
$c_i(v) \Rightarrow{u_i(v)}$
\hspace{2mm} 
\begin{game}{2}{2}
     \> $s_{C}$   \> $\neg s_{C}$\\
$s_{R}$ \> {${-3,-3}$}\>$0,-5$\\
$\neg s_{R}$   \> {$-5,6$}   \>{$5,5$}
\end{game}\hspace*{\fill}%
\caption{From a boolean game to its induced strategic game. Each player $i$ is endowed with boost factor $\mu_i$.}
\label{trans-bool}
\end{figure}

When constructing a strategic game starting from a given boolean game we soon realize the downside of the two main restrictions --- otherwise extremely desirable from a computational point of view --- differentiating boolean games from the larger class of strategic game with goals:

\begin{itemize}

\item The number of strategies agents can play are always, $2^{n}$ for some natural number $n$, which can vary for each agent;

\item The relative distance between goal states and non-goal states is given by the $\mu+1$ boost factor, which is fixed for each agent.

\end{itemize}

The first constraint is not particularly restrictive, but it shows an intuitive difference between boolean games and strategic games with goals: while in boolean games players fully control propositional variables --- i.e., they can always decide whether to set them to true or false --- strategic games with goals can be thought of as a sort of boolean games with {\em admissible} valuation functions, and thereby more general structures. 
The second constraint is somewhat more restrictive, as the choice of the $\mu+1$ boost factor is extremely committal and bears a number of consequences especially if we take utility to be transferable, as we do in our framework. $\mu+1$ is a {\em regret-based} (or even {\em pride-based}) boost factor: the worse non-goal states are for an agent, the higher the payoff at his goal states; the factor is independent of the actual numerical value assigned by the cost function, i.e., it does not vary along with the absolute values of its domain; what is more, it is fixed for all players, i.e., all players apply exactly the same distance to separate goal states and non-goal states.

These are among the reasons why we think that the more general approach allowed by strategic games with goals is in order, which incorporates the important features of the $\mu+1$ class but leaves also space for more variety.

\subsection{Representation results}

The following results establish correspondences between strategic games, strategic games with goals and boolean games.

\begin{proposition}[Strategic games with goals and strategic games]\label{prop:static1}

\textcolor{white}{s}
\begin{enumerate}

\item Let $(\mathcal{S},\{G_i\}_{i\in N})$ with $\mathcal{S} = (N, \{\Sigma_i\}_{i \in N}, \pi)$ be a strategic game with goals and $\omega$ a profile of boost factors. There exists a strategic game $\mathcal{S}^{\prime} = (N, \{\Sigma_i\}_{i \in N}, \pi^{\prime})$  such that, for each $\sigma \in\prod_{i\in N} \Sigma_i$ and each $i \in N$, we have that  $u^{(\mathcal{S},\{G_i\}_{i\in N})(\omega)}_i(\sigma) = \pi^{\prime}(\sigma)$.

\item Let $\mathcal{S}= (N, \{\Sigma_i\}_{i \in N}, \pi)$ be a strategic game. Then there exists a strategic game with goals $(\mathcal{S}^{\prime},\{G\}_i)$ with $\mathcal{S}^{\prime} = (N, \{\Sigma_i\}_{i \in N}, \pi^{\prime})$ such that, for all profiles of boost factors $\omega$, for each $\sigma \in\prod_{i\in N} \Sigma_i$ and each $i \in N$,  $u^{(\mathcal{S},\{G_i\}_{i\in N})(\omega)}_i(\sigma) = \pi^{\prime}(\sigma)$.

\end{enumerate}

\end{proposition}

\begin{proof}

For the first item, start out with a strategic game with goals $(\mathcal{S},\{G_i\}_{i\in N})$, with $\mathcal{S} = (N, \{\Sigma_i\}_{i \in N}, \pi)$, and a boost factor profile $\omega$. The strategic game $\mathcal{S} = (N, \{\Sigma_i\}_{i \in N}, \pi^{\prime})$ with $u^{(\mathcal{S},\{G_i\}_{i\in N})(\omega)}_i(\sigma) = \pi^{\prime}(\sigma)$ for each $\sigma \in\prod_{i\in N} \Sigma_i$ and each $i \in N$ is immediate to construct.

For the second item, consider the strategic game  $\mathcal{S}= (N, \{\Sigma_i\}_{i \in N}, \pi)$ and construct the strategic game with goals $(\mathcal{S}^{\prime},\{G\}_i)$ with $\mathcal{S}^{\prime} = \mathcal{S}$ and for each $i \in N$, set $G_i = \emptyset$. It follows  that for all profiles of boost factors $\omega$ we have that for each $\sigma \in\prod_{i\in N} \Sigma_i$ and each $i \in N$,  $u^{(\mathcal{S},\{G_i\}_{i\in N})(\omega)}_i(\sigma) = \pi^{\prime}(\sigma)$.
\end{proof}

\begin{proposition}[Strategic games and boolean games]\label{prop:static2}

\textcolor{white}{s}
\begin{enumerate}
\item Let $\mathcal{B} = (N, \Phi, c, \{\gamma_i\}_{i\in N}, \{\Phi_i\}_{i\in N}) $ be a boolean game. Then there exist a strategic game $\mathcal{S} = (N, \{\Sigma_i\}_{i \in N}, \pi)$ and a bijective function $f:V \to \prod_{i \in N} \Sigma_i$ such that for all  $i, u^{\mathcal{B}}_i(v) = \pi_i(f(v))$.

\item Let  $\mathcal{S} = (N, \{\Sigma_i\}_{i \in N}, \pi)$ be a strategic game such that, for each $i\in N$ there is $n \in \mathbb{N} \setminus \{0\}$ with $|\Sigma_i|=2^{n}$ . Then there exist a boolean game $\mathcal{B} = (N, \Phi, c, \{\gamma_i\}_{i\in N}, \{\Phi_i\}_{i\in N}) $, a $k \in \mathbb{N}$, and a bijective function $f:V \to \prod_{i \in N} \Sigma_i$ such that for all  $i, u^{\mathcal{B}}_i(v) = \pi_i(f(v)) - k$.
\end{enumerate}
\end{proposition}

\begin{proof}

For the first item, simply construct a strategic game such that $f$ is a bijection and then set each $\pi_i(f(v))$ to return the value given by $u^{\mathcal{B}}_i(v)$. For the second one, pick again $f$ to be bijection, and, for a sufficiently large $k\in \mathbb{N}$, set each $\gamma_i$ to $\bot$ --- i.e., $p\wedge \neg p$ for some $p \in \Phi_i$ --- and each $u^{\mathcal{B}}_i(v)$ to $\pi_i(f(v))-k$. \end{proof}

\begin{proposition}[Strategic games with goals and boolean games]\label{prop:static3}

\textcolor{white}{s}

\begin{enumerate}
\item Let $\mathcal{B} = (N, \Phi, c, \{\gamma_i\}_{i\in N}, \{\Phi_i\}_{i\in N}) $ be a boolean game. Then there exist a strategic game with goals $(\mathcal{S} = (N, \{\Sigma_i\}_{i \in N}, \pi), G_i)$ and a bijective function $f:V \to \prod_{i \in N} \Sigma_i$ such that for all  $i, u^{\mathcal{B}}_i(v) = u_i(f(v))$.

\item Let  $(\mathcal{S} = (N, \{\Sigma_i\}_{i \in N}, \pi), G_i)$ be a strategic game such that, for each $i\in N$ there is $n \in \mathbb{N} \setminus \{0\}$ with $|\Sigma_i|=2^{n}$ . Then there exist a boolean game $\mathcal{B} = (N, \Phi, c, \{\gamma_i\}_{i\in N}, \{\Phi_i\}_{i\in N}) $, a $k \in \mathbb{N}$, and a bijective function $f:V \to \prod_{i \in N} \Sigma_i$ such that for all  $i, u^{\mathcal{B}}_i(v) = u_i(f(v)) - k$.
\end{enumerate}

\end{proposition}

\begin{proof}
Direct consequence of the definitions and the previous two results.
\end{proof}

All in all, strategic games with goals and strategic games display a straightforward correspondence. Modulo some minor requirements, this is also true for boolean games. Moreover, the proofs allow us to talk about the boolean/strategic game (with goals) {\em corresponding } to a boolean/strategic game (with goals), and provide automatic procedures to translate among these classes.

The reader might at this stage understandably be puzzled. We have introduced strategic games with goals, an alleged generalization of strategic games, showing that, in the end, the two structures are substantially equivalent. The rest of the paper, in what we believe is its main contribution, is devoted to showing that the introduction of dynamic operations, such as the possibility of side-payments in a pre-play phase, brings to light the striking differences between these structures.

\section{Endogenous games and quasi-dichotomous preferences}\label{sec:endo}

Endogenous games \cite{JW05} have been introduced as an extension of normal form games with a pre-play negotiation phase, where players have the possibility, before the game starts, to spend the amount of utility received at certain outcomes to influence their opponents' decision-making.

This section imports the ideas and the techniques of endogenous games to games with quasi-dichotomous preferences and it shows  that, in spite of the static representation results of the previous section, the presence of goal states does make a difference when side-payments are allowed.

Following \cite{JW05} \footnote{While the approach to game transformation by side-payments adopted here \cite{JW05} provides an elegant technical framework that well suits strategic games with goals, we  remind the reader of the existence of earlier related work in the game theory literature \cite{gut1}, \cite{kal1}, \cite{farrell}. Also, a more involved model of pre-play negotiations in non-cooperative games, overcoming a number of limitations of \cite{JW05}, has recently been developed \cite{GT-manifestoArxiv}.
} we enrich a strategic game  with a family  $\{T_i\}_{i\in N}$ where, for each player $i$, $T_i$ is a set of functions of the form $\tau_i: \Sigma \times N \to \mathbb{R}_+$, such that
 $\tau_i(\sigma,i) = 0$ for each $\sigma \in \prod_{i \in N} \Sigma_i$. Each such function specifies how much payoff player $i$ secures the other players in case some particular outcome obtains. We call each $\tau \in \prod_{i \in N} \tau_i$ a {\bf transfer function} 
and a tuple $(\mathcal{S}, \{T_i\}_{i\in N})$ an {\bf endogenous game}. 
An {\bf endogenous game with goals} is defined in the expected way, i.e., as a tuple $((\mathcal{S}, \{G_i\}_{i\in N}), \{T_i\}_{i\in N})$ where $(\mathcal{S}, \{G_i\}_{i\in N})$ is a strategic game with goals and $(\mathcal{S}, \{T_i\}_{i\in N})$ is an endogenous game.

It is useful to think of an endogenous game (with goals) as a game consisting of two phases:

\begin{itemize}

\item A pre-play phase, where players simultaneously decide on their transfers to the other players;

\item An actual game play, where the utility of the starting game is updated taking the selected transfers into account.

\end{itemize}

\begin{definition}[Update by side-payments]

Let $\mathcal{S}=(N, \{\Sigma_i\}_{i \in N}, \pi)$ be a strategic game and let $\tau$ be a transfer function. The {\bf play of $\tau$ in $\mathcal{S}$} is the strategic game $\tau(\mathcal{S})= (N, \{\Sigma_i\}_{i \in N}, \pi^{\prime})$ where, for $\sigma \in \Sigma, i \in N$, 

\begin{eqnarray*}
  \pi^{\prime}_i(\sigma)  = \pi_i(\sigma)  {}  + \sum_{j\in N} \tau_j(\sigma,i)  - \sum_{j \in N}\tau_i(\sigma,j) 
\end{eqnarray*}

\end{definition}

In a nutshell, when a game is updated by a transfer function each player:

\begin{itemize} 

\item adds to his payoff at each outcome all the transfers that he receives from the other players at that outcome;

\item subtracts from his payoff at each outcome all the transfers he makes to the other players at that outcome.

\end{itemize}


Let $\mathcal{E}=((\mathcal{S}, \{G_i\}_{i\in N}), \{T_i\}_{i\in N})$ be an endogenous game with goals and let $\mathcal{E}(\omega)$ be its instantiation with the boost factor profile $\omega$\footnote{The calculation of the utility after a transfer function has occurred is carried out as expected, i.e., in the updated strategic games with goals instantiated by $\omega$.}. A pair $(\tau, \sigma)$, for $\tau$ being a transfer function and $\sigma$ a strategy profile available at $\mathcal{S}$, is a {\bf solution} of $\mathcal{E}$ if there is a strategy in the two-phase game that is a subgame perfect equilibrium and where $(\tau, \sigma)$ is played on the equilibrium path. 
In other words in order for a pair $(\tau, \sigma)$ to be a solution of an endogenous game with goals $\mathcal{E}= (\mathcal{S}, \{G_i\}_{i\in N}), \{T_i\}_{i\in N})(\omega)$ we require that there exists a strategy for $\mathcal{E}(\omega)$, i.e., a specification of a strategy profile $\sigma^{\prime}$ for the actual game play after every transfer function $\tau^{\prime}$, such that: \begin{enumerate}

\item $(\tau, \sigma)$  is a Nash Equilibrium in the two-phase game;

\item  for every $(\tau^{\prime},\sigma^{\prime})$ in the strategy specification, $\sigma^{\prime}$ is a Nash-equilibrium  in every subgame $\tau^{\prime}(\mathcal{S}, \{G_i\}_{i\in N})(\omega)$ of the two-phase game. 

\end{enumerate}

Subgame perfect equilibria rule out {\em incredible threats} \cite{OR}, in our case the fact that some transfers might be discouraged by the play of strategies that are dominated after the transfers in question are made. 

For a solution $(\tau, \sigma)$ of $\mathcal{E}= ((\mathcal{S}, \{G_i\}_{i\in N}), \{T_i\}_{i\in N})$ instantiated with $\omega$ we say that $\sigma$ is a {\bf surviving equilibrium} if it is a Nash equilibrium of $(\mathcal{S}, \{G_i\}_{i\in N})(\omega)$.

We also impose, for purposes that will result clear later on, that strategies are {\em uniform} modulo positive affine transformations  \cite{OR}, i.e., players make the same choices in two games $\mathcal{S}$ and $\mathcal{S}^{\prime}$ where $\mathcal{S}^{\prime}$ is obtained from $\mathcal{S}$ via positive affine transformation. In other words, players play the same strategy in games that are {\em practically identical}. We call this requirement {\em strategy uniformity}, an example of which is given in Figure \ref{strategy-uniformity}.

\begin{figure}[htb]\hspace*{\fill}%
\begin{game}{2}{2}
    \> $L$  \> $R$\\
$U$ \> ${2,2}$\>$0,0$\\
$D$   \> $0,0$   \>$1,1$
\end{game}\hspace*{\fill}%
\hspace{2mm}  
\begin{game}{2}{2}
     \> $L$   \> ${R}$\\
$U$ \> {$3,3$}\>$1,1$\\
$D$   \> {$1,1$}   \>{$2,2$}
\end{game}\hspace*{\fill}%
\caption{Strategy uniformity: if a strategy profile is chosen in the game on the left side, then there is no reason for players to deviate from this choice in the game to the right side.}
\label{strategy-uniformity}
\end{figure}

\subsection{Boost factors and game update}

When a pre-play phase precedes a strategic game with goals instantiated with a family of boost factors, we need to understand how the boost factors react to the game update. In particular, if boost factors are to encode a relative distance between goal states and non-goal states, we do want these changes to be reflected in the dynamics introduced by the possibility of side-payments.

To do so, we indicate $\omega^{(\mathcal{S},\{G_i\}_{i\in N})}$ the boost factor applied to game $(\mathcal{S},\{G_i\}_{i\in N})$ and omit the superscript when obvious. We indicate with $\Omega\uparrow$ the set of all boost factors satisfying the following property, under the assumption that $\pi$ is the payoff function of $\mathcal{S}$ and $\pi^{\prime}$ of $\mathcal{S}^{\prime}$:

\begin{equation}
\omega_i^{(\mathcal{S},\{G_i\}_{i\in N})} (x) \geq \omega_i^{(\mathcal{S}^{\prime},\{G_i\}_{i\in N})}(x) \Leftrightarrow \exists \sigma \not\in G_i \\ \mbox{ such that } \forall \sigma^{\prime} \not\in G^{\prime}_i,  \pi_i(\sigma) \geq \pi^{\prime}_i(\sigma^{\prime})
\end{equation}

What the definition says is that, fixing a boost factor, the utility of goal states is pushed upwards the higher the payoff non-goal states yield. It is pushed downwards otherwise.
 This property, as well as the following other, becomes extremely relevant for our purposes when $\mathcal{S}^{\prime} = \tau(\mathcal{S})$, for a given transfer function $\tau$.

We will also consider a different type of boost factor, modelling players that value more goal states the further away they are from the worst possible outcome they could end up in. It is the {\em regret-based} boost factor that generalizes the $\mu$ factor typical of boolean games \cite{mike}. We indicate with $\Omega\downarrow$ the set of all boost factor profiles satisfying the following property: 

\begin{equation}
\omega_i^{(\mathcal{S},\{G_i\}_{i\in N})} (x)\geq \omega_i^{(\mathcal{S}^{\prime},\{G_i\}_{i\in N})}(x) \Leftrightarrow \exists \sigma \not\in G_i \\ \mbox{ such that } \forall \sigma^{\prime} \not\in G^{\prime}_i,  \pi_i(\sigma) \leq \pi^{\prime}_i(\sigma^{\prime})
\end{equation}

What the definition says is that, fixing a boost factor, the utility of goal states is pushed upwards the lower the payoff non-goal states yield. It is pushed downwards otherwise.

As we will see later, having a boost factor profile in the set $\Omega\downarrow$ makes a huge difference in players' strategic behaviour, as players can deliberately strive to increase their cost at non-goal states for the sole reason of increasing their payoff at goal states.

Unless otherwise specified we consider boost factor profiles in the set $\Omega\uparrow$.

\subsection{Budget constraints}

We also consider strategic games with goals where players have a budget constraint, i.e., a limit to the amount of payoff they can end up gaining as a result of the pre-play phase. Intuitively, for $\mathcal{S}=(N, \{\Sigma_i\}_{i \in N}, \pi)$  being a strategic game, each player $i$ is associated a budget constraint $b_i : \Sigma \to \mathbb{R}$, with $b_i(\sigma) \geq \pi_i(\sigma)$, for each $\sigma \in \Sigma$. 

Clearly many transfer functions exist that, applied to a starting strategic game, violate a given budget constraint for some player. Rather than declaring them outright inadmissible and ruling them out from the model, we introduce a uniform punishment to all players if such transfer functions are played. By doing so, we will still be able to talk about violation of the budget constraint requirement and of situations in which players can {\em rationally} decide to rule out such profiles. 

The definition of utility of endogenous games instantiated with a boost factor profile $\omega$ and budget constraints $\{b_i\}_{i\in N}$ is given below. To avoid clutter of notation we denote $(\mathcal{S},G)(\omega, b^{\mathcal{S}})$ the instantiation of strategic game  $\mathcal{S}$ endowed with the family of goals $\{G_i\}_{i\in N}$ boost factors $\{\omega_i\}_{i\in N}$ and the family of budget constraints $\{b^{\mathcal{S}}_i\}_{i \in N}$ defined on $\mathcal{S}$ and omit superscripts and subscripts whenever possible.

\begin{definition}[Utility with budget constraints]\label{def:utilities+bc} Let $(\mathcal{S^{\prime}},G)(\omega, b)$ be an instantiated strategic game with goals and $\epsilon \in \mathbb{R}_+$ a sufficiently small positive real. Let moreover $\mathcal{S} = \tau(\mathcal{S}^{\prime})$ for some strategic game $\mathcal{S^{\prime}}$ and transfer function $\tau$, and let $\pi$ the payoff function of $\mathcal{S}$. The {\bf utility function} $u^{(\mathcal{S},G)(\omega, b^{\mathcal{S^{\prime}}})}$ is defined as follows.

$
u^{(\mathcal{S},G)(\omega, b^{\mathcal{S^{\prime}}})}_i(\sigma)= \left\{
\begin{array}{ll}
\omega(\pi_i(\sigma)) - |D|\kappa & \IF \sigma \in G_i\\
\pi_i(\sigma) - |D|\kappa  & \mbox{otherwise}
\end{array}
\right.
$

\vspace{1cm}
where $\kappa = max \{0 , -min \{\pi_i(\sigma) - b_i(\sigma)  \mid \sigma \in \Sigma \mbox{ and } i\in N \}\}$ and $|D| = \{ i\in N \mid (\pi_i(\sigma) - b_i(\sigma))= \kappa \mbox{ for some } \sigma \in \Sigma\}$ 
\end{definition}

The difference with Definition \ref{def:utilities} is the subtraction of the $\kappa$ factor to the players' payoff. This factor makes a difference only if there exists some player and some outcome where the player exceeds his budget constraints as a result of the pre-play phase. Then $\kappa$ is given the value of the {\em highest excess in the game}, i.e., the highest amount of payoff that a player has exceeded at some outcome. $\kappa$ is multiplied by the number of players $|D|$ that simultaneously exceed the threshold by this highest amount.

At first one might think that players can exceed their budget constraint strategically, i.e., their might be situations in which players might find it rational to play strategies that deliberately go beyond the limit imposed by their constraint for instance in order to punish the other players. The following proposition shows that this is never an equilibrium strategy.

\begin{proposition}

Let $x = (\tau,\sigma)$ be a subgame perfect equilibrium of the endogenous game with goals $\mathcal{E} = ((\mathcal{S}, \{G_i\}_{i \in N}), \{T_i\}_{i\in N})(\omega)$  with budget constraints $\{b_i\}_{i\in N}$ on $\mathcal{S}$. Then the punishment factor of $(\tau(\mathcal{S}), \{G_i\}_{i \in N}))(\omega)$ equals 0.

\end{proposition}

\begin{proof}

Suppose not. Then there exist a set of players $I$ and a strategy profile $\sigma^{*}$ where, for each $i \in I$, $\pi_i(\sigma^{*}) - b_i(\sigma^{*}) = \kappa > 0$ holds in $\tau(\mathcal{S})$. Fix now a player $j\not\in I$ and let $\kappa_j<\kappa$ the maximal excess of player $j$ at $\tau(\mathcal{S})$. Now, by strategy uniformity, player $i$ is better off deviating to a transfer function $(\tau^{*}_i,\tau_{-i})$ where $\tau^{*}_i(j,\sigma^{\prime}) = \tau_{i}(j,\sigma^{\prime})+ \frac{\kappa - \kappa_j}{2}$ for each $\sigma^{\prime}$ as, notice,  $(\tau^{*}_i,\tau_{-i})(\mathcal{S})$ can be obtained from $\tau(\mathcal{S})$ via positive affine transformation and the punishment factor in $(\tau^{*}_i,\tau_{-i})(\mathcal{S})$ is strictly smaller than $|D|k$ . As a consequence, the play of $\sigma$ in $(\tau^{*}_i,\tau_{-i})(\mathcal{S})$ yields player $i$ a strictly better payoff.

\end{proof}

As we will see next, budget constraints allow for a neat connection with the endogenous boolean games studied in \cite{EBG}.

\paragraph{Dynamic asymmetries}

The reader may have wondered what the difference actually is between strategic games with goals and strategic games. After all, we have shown a rather straightforward correspondence between the two structures in Section 2. 

The following proposition is a hint of the fact that the similarity between them is not as obvious as the previous results might suggest. Concretely, when introduce dynamic operations that update the game (such as side-payments), the correspondence ceases to exist.

\begin{proposition}[Dynamic asymmetry]\label{prop:dynamic}

Let $(\mathcal{S},\{G_i\}_{i\in N})$ with $\mathcal{S} = (N, \{\Sigma_i\}_{i \in N}, \pi)$ be a strategic game with goals, $\omega$ a boost factor profile, $\tau^{0}, \tau^{\prime}$ be two  transfer functions, $\mathcal{S}^{0},\mathcal{S}^{\prime}$ the strategic games corresponding to $\tau^{0}(\mathcal{S},\{G_i\}_{i\in N})(\omega)$ and $\tau^{\prime}({\mathcal{S},\{G_i\}_{i\in N})}(\omega)$, respectively. It is not necessarily the case that $\tau^{\prime}(\mathcal{S}^{0}) = \mathcal{S}^{\prime}$.


\end{proposition}

\begin{proof}

Consider the strategic game with goals $(\mathcal{S},\{G_i\}_{i\in N})$ where $\mathcal{S}=(\{A,B\}, \{\Sigma_i\}_{i \in N}, \pi)$, $G_B = \emptyset$, $G_A = \{\sigma^{*}_A,\sigma^{*}_B\}$, $\pi_i(\sigma) = 0 $ for all $i\in N, \sigma \in \Sigma$ and $|\Sigma_{A}|=|\Sigma_B| = 2$. Consider the transfer functions $\tau^{0}$ and $\tau^{\prime}$ such that $\tau^{\prime}_A(\sigma,B) = 2$, for all $\sigma  \not\in G_A$ while  $\tau^{\prime}_A(\sigma^{*},B) = 0$ and  $\tau^{\prime}_B(\sigma,A) = 0$ for all $\sigma \in \Sigma$. Notice now $u^{(\tau^{0}(\mathcal{S}),\{G_i\}_{i\in N})(\omega) }_A(\sigma^{*}) = \omega^{(\tau^{0}(\mathcal{S}),\{G_i\}_{i\in N})}_A(0)>0$, that $u^{(\tau^{\prime}(\mathcal{S}),\{G_i\}_{i\in N})(\omega)}_A(\sigma^{*}) = \omega^{(\tau^{\prime}(\mathcal{S}),\{G\}_i)(\omega)}_A(-2)>-2$ and that $\omega^{(\tau^{\prime}(\mathcal{S}),\{G\}_i)(\omega)}_A(-2) < 0$ . $\sigma^{*}$ yields player $A$ a payoff $x \geq 0$ in $\tau^{\prime}(\mathcal{S}^{0})$ but not in $\mathcal{S^{\prime}}$.
 \end{proof}

From the last proposition we gather that updating strategic games with goals by means of side-payments--- but a similar argument extends to boolean games with incentives \cite{mike} in a straightforward fashion --- destroys the direct correspondence with strategic games. The reason lies on the fact that the utility function is calculated in such a way that a player gets a boosted payoff in outcomes satisfying his goal, but always relative to the value yielded by the non-goal states. 
The results in the next section will show that this imbalance bears heavy consequences in terms of equilibrium analysis.

\subsection{Relation of side-payments with payoffs and their dynamics}

This part looks at the relation between side-payments and their effect on payoffs, zooming in the relation with taxation mechanisms, of the type studied in the related literature \cite{mike}.

\paragraph{Side-payments and payoffs}

\begin{proposition}[Normalization]

Let $(\mathcal{S},\{G_i\}_{i\in N})$ with $\mathcal{S} = (N, \{\Sigma_i\}_{i \in N}, \pi)$ be a strategic game with goals, $\omega$ be a profile of boost factors, $\{b_i\}_{i\in N}$ a family of budget constraints on $\mathcal{S}$ and $\tau$ a transfer function on $\mathcal{S}$. Then there is a strategic game with goals $(\mathcal{S}^{\prime},\{G_i\}_{i\in N})$ with $\mathcal{S}^{\prime} = (N, \{\Sigma_i\}_{i \in N}, \pi^{\prime})$ such that, for all $i \in N$ and $\sigma \in \Sigma$, we have that $u^{(\tau(\mathcal{S}),\{G_i\}_{i\in N})(\omega,\{b_i\}_{i\in N})}_i(\sigma)=u^{(\tau^{0}(\mathcal{S}^{\prime}),\{G_i\}_{i\in N})(\omega, \{b_i\}_{i\in N})}_i(\sigma)$.

\end{proposition}

\begin{proof} Let  $\mathcal{S} = (N, \{\Sigma_i\}_{i \in N}, \pi)$, $\tau$ be a transfer function on $\mathcal{S}$, $\omega$ be a profile of boost factors, and let  $\kappa$ be the punishment factor associated to $\tau(\mathcal{S})$ by the family of budget constraints $\{b_i\}_{i\in N}$. Let $\mathcal{S} = (N, \{\Sigma_i\}_{i \in N}, \pi^{\prime})$ be a game different from $\mathcal{S}$ only in the payoff function 
and let $\pi^{\prime}_i(\sigma)=\pi_i(\sigma) + \sum_{j\in N} \tau_j(\sigma,i)  - \sum_{j \in N}\tau_i(\sigma,j)  - \kappa$. We have that for all $i\in N$ and $\sigma \in \Sigma$,   $u^{(\tau(\mathcal{S}),\{G_i\}_{i\in N})(\omega,\{b_i\}_{i\in N})}_i(\sigma)=u^{(\tau^{0}(\mathcal{S}^{\prime}),\{G_i\}_{i\in N})(\omega, \{b_i\}_{i\in N})}_i(\sigma)$.\end{proof}

In a nutshell, we can think of transfer functions as payoff modifiers. As we will see, the setting of strategic games with goals allows to explicitly strategise over payoffs by a rational use of side-payments, rather than taking them as immutable factor in a game.

\paragraph{Side-payments and taxes}

Wooldridge et al. \cite{mike} define taxation mechanisms on a boolean game $\mathcal{B}=(N, \Phi, c, \{\gamma_i\}_{i\in N}, \{\Phi_i\}_{i\in N})$ as functions $\alpha_i: V \to \mathbb{R}_+$ where every player receives a monetary {\em sanction} at each particular outcome. 

We employ this notion for the more general framework of strategic games with goals, defining them as functions $\alpha_i: \Sigma \to \mathbb{R}_+$, subtracting the taxes received from a player at a certain outcome to his payoff.

To ease the connection with our previous definitions we introduce the strategic game  $\alpha(\mathcal{S})$, i.e., the  game $\mathcal{S}$ to which the taxation mechanism $\alpha = \prod_{i\in N}\alpha_i$ is applied, as the strategic game $(N, \{\Sigma_i\}_{i\in N}, \pi^{\prime})$ where for each $i\in N,\sigma \in \Sigma$ we have that $\pi^{\prime}_i(\sigma) = \pi_i(v)-\alpha_i(\sigma)$.

Transfer functions and taxation mechanisms are of a rather different kind. While the former ones consist of payoff redistributions among players at certain outcomes, without adding or subtracting to the players' total payoff, the latter ones explicitly inject new sanctions into the system to modify players' decision-making. 
However, in a technical sense, we can always find a taxation mechanism having the same effect of a transfer function on an underlying game.

\begin{proposition}[From side-payments to taxes]\label{prop:reduction}

Let $(N, \{\Sigma_i\}_{i\in N}, \pi^{\prime})$ be a strategic game and $\tau$ a transfer function. There exists a taxation mechanism $\alpha$ such that $NE(\tau(\mathcal{S})) = NE(\alpha(\mathcal{S})) $.
\end{proposition}

\begin{proof}
Straightforward.
\end{proof}

Proposition \ref{prop:reduction} shows that, in some respect, taxation mechanisms can {\em simulate} transfer functions. 
It must however be said that simulations of this kind only make sense when both taxation mechanisms and transfers functions are fixed. But unlike taxation mechanisms, that are decided externally, transfer functions bear further strategic considerations. As made clear in the motivating example, it is not enough to establish that a transfer function induces equilibria in the resulting game, but we also need to establish whether the transfer function itself is part of a larger equilibrium, i.e., whether players are not better off by switching to different transfers. 

\subsection{Endogenous Boolean Games}

As we have observed, preference quasi-dichotomy in boolean games is modelled using a {\em boost} factor $\mu_i$ for states satisfying player $i$'s goal, which gets his utility increased adding to the payoff already associated to each such state the cost and the taxes he would get at his worst possible outcome. As costs and taxes are never negative, this leads players' utility to be maximal at goal states.
The $\mu_i$ factor is a classical case of regret-based boost factor --- i.e., it belongs to $\Omega_{\downarrow}$ --- and, as the reader will have noticed, the expression $- \beta_i(v)$, describing the net gain that player $i$ obtains from the transfer function $\beta$ at $v$, can exceed $c_i(v)$, i.e., the cost $i$ is incurring in  at $v$. This means that, using side-payments, players might incur in {\em de facto} negative costs for performing certain actions and, consequently, the boost factor associated to goal states might also be negative, disrupting preference quasi-dichotomy. To avoid this drawback we use budget constraints and, as usual,  we impose that {after a transfer is made, each game undergoes a positive affine transformation}, where each valuation $v$ is assigned an extra correction factor $|D|\kappa$ for $D$ being the amount of players exceeding of the maximum level \footnote{In the original article \cite{EBG} the correction factor was simply $\kappa$. This further modification allows to make players individually responsible for their simultaneous maximum excess and to break ties. In any case, the results of \cite{EBG} all carry over to this framework.}.

We now define an endogenous boolean game as a tuple $(\mathcal{B}, \{T_i\})$, where $\mathcal{B}$ is a boolean game and $T_i$ is a set of transfer functions for player $i$, fixing both the boost factors --- i.e., $\mu_i$ --- and a family of budget constraints $\{b_i\}_{i \in N}$ such that  $b_i(v) = - c_i(v)$ for each player $i$ and outcome $v$, omitting them when possible to ease reading.

Now we can define a quasi-dichotomous utility function for boolean games with side-payments.

\begin{definition}[Utility]\label{def:utilities} Let $\mathcal{B}$ be a boolean game, $\beta$ a (boolean) transfer function and $\epsilon \in \mathbb{R}_+$ a sufficiently small positive real. The {\bf utility function} $u^{\beta(\mathcal{B})}: N \times V \to \mathbb{R}$ assigning to each player the payoff $u^{\beta(\mathcal{B})}_i(v)$ he receives at valuation $v$ is defined as follows.

$
u^{\beta(\mathcal{B})}_i(v)= \left\{
\begin{array}{ll}
\epsilon + \mu_i  - (c_i(v) + \beta_i(v) + \kappa)& \IF v \models \gamma_i\\
- (c_i(v) + \beta_i(v)  + \kappa) & \mbox{otherwise}
\end{array}
\right.
$
\end{definition}

where $\kappa$ is defined as in the general case.

The definition basically says that at outcomes --- i.e., valuations --- satisfying players' goals, players sum to their costs and the appropriately corrected transfers made also an extra factor $\epsilon+\mu_i$. For the other outcomes instead, players receive their costs plus the appropriately corrected transfers made.
The definition shows the specificity of boolean games, where players are cost-minimizers but never favour a cheap choice that does not reach a goal to an expensive one that instead does. Therefore, the expression $u_i(v)$ is always strictly positive in case $v\models \gamma_i$ and never strictly positive otherwise. 

Using the results in \cite{mike} we can transfer the following results to our framework.

\begin{proposition}[Complexity of finding equilibria]
Let $\mathcal{B} = (N, \Phi, c, \{\gamma_i\}_{i\in N}, \{\Phi_i\}_{i\in N})$ be a boolean game and $\phi$ a boolean formula constructed on $\Phi$. The following hold:

\begin{enumerate}

\item The problem of verifying whether, for $\beta \in \prod_{i\in N} B_i$ and $v\in V$, $v\in NE(\beta(\mathcal{B}))$ is {\textit{co-NP complete}};
\item The problem of verifying whether some boolean transfer function $\beta$ and outcome $v \in NE(\beta(\mathcal{B}))$ exist such that $v \models \phi$ is {\textit{$\Sigma^{2}_p$ complete}};
\item The problem of verifying whether some boolean transfer function $\beta$ exists such that $NE(\beta(\mathcal{B})) \neq \emptyset$ and for all $v^{\prime} \in NE(\beta(\mathcal{B}))$ we have that $v^{\prime} \models \phi$ is {\textit{$\Sigma^{2}_p$ complete}}.
\end{enumerate}

\end{proposition}

\begin{proof}

The proofs of \cite[Proposition 1]{mike}, \cite[Proposition 6]{mike}, \cite[Proposition 14]{mike} immediately carry over to our case. \end{proof}

The following section carries out an equilibium analysis of endogenous games with goals in the full-blown two-phase game.

\section{Equilibrium analysis}\label{sec:equilibrium}



The classical results on equilibrium survival for the case of endogenous games \cite{JW05} are centred on the notion of solo payoff, i.e., the payoff that a single player $i$ can guarantee if the opponents $-i$ do not make any transfer. 

\begin{definition}[Solo payoff]\label{def:solo}

Let $\mathcal{E}(\omega)= ((\mathcal{S},\{G_i\}_{i\in N}), \{T_i\}_{i\in N})(\omega)$ be an endogenous game with goals instantiated by a boost factor profile $\omega$. The {\bf solo payoff} $\hat{s}_i$ for player $i$, where {sup} is the least upper bound operator, is

$$\hat{s}^{\mathcal{E}(\omega)}_i= \mbox{sup}_{\tau_i} (min_{\rho \in NE^{\Delta}((\tau_i, {\tau^0}_{-i})(\mathcal{E}(\omega))} E^{(\tau_i, {\tau^0}_{-i})(\mathcal{E}(\omega))}(\rho))$$

\end{definition}

In words, the solo payoff of player $i$ is given by the best transfer $i$ can make, under the expectation that his opponents will not make any transfer and will play the worst for $i$ Nash equilibrium in each subgame. When the underlying game is fixed and no confusion can arise, we use the notation $\hat{s}_i$.

Two important facts are known for endogenous games  \cite{JW05}:

\begin{enumerate}

\item when $N\leq 2$, a Nash equilibrium survives if and only if each player gets at least his solo payoff;
\item when $N>2$, every pure strategy Nash equilibrium survives.

\end{enumerate}

It is not obvious at all that the stability results carry over to our case, due to the elaborated construction of the utility function (Definition \ref{def:utilities}).
In fact we can show that the second statement does not hold for endogenous games with goals.

\begin{proposition}\label{prop:nonsurvival}

There exists an endogenous game with goals $\mathcal{E}(\omega)$ with more than $2$ players and Nash equilibrium outcome $\sigma$ of $\mathcal{E}(\omega)$ that is not a surviving equilibrium.

\end{proposition}

\begin{proof}

Consider the strategic game with goals $(\mathcal{S},\{G_i\}_{i\in N})$ where $\mathcal{S}=(\{A,B,C\}, \{\Sigma_i\}_{i \in \{A.B,C\}}, \pi)$, and where $G_B = G_C = \emptyset$, $G_A = \{\sigma^{*}_A,\sigma^{*}_B, \sigma^{*}_C\}$, $\pi_i(\sigma) = 0 $ for all $i\in N, \sigma \in \Sigma$. 
Clearly each $(\sigma^{*}_A,\sigma^{*}_B, \sigma^{\prime}_C)$ for $\sigma^{\prime}_C \neq \sigma^{*}_C$ is a Nash equilibrium. However it is not a surviving one. For suppose it was and take $(\tau, (\sigma^{*}_A,\sigma^{*}_B, \sigma^{\prime}_C))$ to be the strategy played on the equilibrium path. Now for each payoff that player $C$ is getting at $(\tau, (\sigma^{*}_A,\sigma^{*}_B, \sigma^{\prime}_C))$ player $A$ is better off deviating to a transfer $\tau^{\prime}_A$ such that $\tau^{\prime}_A(C,\sigma^{*}) >  \pi^{\tau}_C(\sigma^{*}_A,\sigma^{*}_B, \sigma^{\prime}_C) - \pi^{\tau}_C(\sigma^{*}_A,\sigma^{*}_B, \sigma^{*}_C)$, making it worthwhile for $C$ to satisfy $A$'s goal.
\end{proof}

The idea of the proof is quite simple: when goal realization is at stake, players could go to any length to have it satisfied. So Nash equilibria of the initial game will not survive if a joint deviation of a group of players satisfies the goal of some other player without compromising their own. Figure \ref{nonsurvival:goals} illustrates one more such scenario.

\begin{figure}[htb]\hspace*{\fill}%
\begin{game}{2}{2}[{$1$}]
     \> {$L$}   \> $R$\\
{$U$} \> \red{ $1,1,1$   }\>{$1,1,1$}\\
$D$   \>{${1,1,1}$}  \>{$1,1,1$}
\end{game}\hspace*{\fill}%
\begin{game}{2}{2}[$2$]
     \>{$L$}   \> $R$\\
{$U$} \> { $1,1,1$   }\>$1,1,1$\\
$D$   \>$1,1,1$  \>{$1,1,1$}
\end{game}\hspace*{\fill}%
\caption{A three player game, with the third player choosing the matrix to be played. \textcolor{red}{Row} wants $(U,L,1)$ to be realized and, no matter what the distribution of payoffs after the pre-play phase looks like, he is willing to compensate the other players to go along that strategy. Notice that \textcolor{red}{Row}'s deviation in the pre-play phase is effective, as it does not compromise the opponents' goals. The Nash equilibrium outcome $(D,R,2)$ is not surviving.}
\label{nonsurvival:goals}
\end{figure}

Proposition \ref{prop:nonsurvival} is of fundamental importance, not only because it shows that there is a boolean game with $|N|>2$ where a pure strategy Nash equilibrium is not surviving, at odds with well-known results for the strategic games case, but because Propositions \ref{prop:static1}, \ref{prop:static2}, \ref{prop:static3} were suggesting a straightforward correspondence between strategic games and strategic games with goals. 
The reason of this imbalance was already hinted at by Proposition \ref{prop:dynamic}, which showed how the application of the same transfer function to a strategic game and to its corresponding strategic game with goals was not guaranteed to yield corresponding structures. Proposition \ref{prop:nonsurvival} uses the same idea to show that dynamic factors such as transfer functions are enough to falsify fundamental results for the otherwise statically correspondent strategic games.

With boost factors in $\Omega \uparrow$, it is also interesting to notice that, for outcomes that are uniquely shared joint goal among players, i.e., for outcomes $\sigma$ for which $\sigma \in \bigcap_{i \in N} G_i$ and for which $(\sigma^{\prime}_{i},\sigma_{-i}) \not\in G_i$ unless $(\sigma^{\prime}_{i},\sigma_{-i}) = \sigma$,  we have that $\sigma$ is always a surviving equilibrium.

\begin{proposition}\label{prop:joint}
Let $\mathcal{E}(\omega)$ be an endogenous game with goals instantiated by a boost factor profile $\omega \in \Omega \uparrow$. Then every outcome $\sigma$, for which  $\sigma \in \bigcap_{i \in N} G_i$ and for which $(\sigma^{\prime}_{i}\sigma_{-i}) \not\in G_i$ unless $(\sigma^{\prime}_{i},\sigma_{-i}) = \sigma$,  we have that $\sigma$ is always a surviving equilibrium.

\end{proposition}

\begin{proof}
Straightforward
\end{proof}

Also, in spite of the negative result discussed above, we are still able to show a sufficient condition for pure strategy Nash equilibria to survive in certain strategic games with goals, independently of the number of players involved. 

\begin{proposition}[Survival]\label{prop:survival2}

Let $(\mathcal{S},\{G_i\}_{i\in N})$ be strategic game with goals and let it be instantiated by a boost factor profile $\omega$. A pure strategy Nash equilibrium $\sigma$ of $\mathcal{S}$ survives whenever $u_i(\sigma) \geq \hat{s}_i$ for each $i\in N$.

\end{proposition}

\begin{proof}

We need to construct a subgame perfect equilibrium of the two-phase game where $\sigma$ is played. On the equilibrium path let players play the profile $(\tau^{0}, \sigma)$. Off the equilibrium path, if a single player deviates from $\tau^{0}$ pick the worst Nash equilibrium for that player and have that played in the continuation. If more than a player deviates from $\tau^{0}$, play any Nash equilibrium. We can observe that no player $i$ can get more than $\hat{s}_i$ by deviating from $\tau^{0}$, given the choices played in each subgame. Moreover $\sigma$ is a pure strategy Nash equilibrium of $\mathcal{S}$. \end{proof}

What we have shown is that when a Nash equilibrium $\sigma$ gives the players at least their solo payoff then that Nash equilibrium survives. In fact the proof of the proposition shows even more, i.e., that the players will actually obtain at least their solo payoff in that equilibrium.

\paragraph{The case of Boolean Games}

For boolean games the results are even more striking

\begin{proposition}[No survival]\label{prop:nonsurvivalbg}

There exist an EBG $(\mathcal{B},\{B_i\}_{i\in N})$ with $|N|=3$ and $\{v\} = NE(\beta^{0}(\mathcal{B}))$ such that:

\begin{itemize}

\item $v \models \gamma_i$, for each $i\in N$
\item $v$ is not a surviving equilibrium.

\end{itemize}

\end{proposition}

\begin{proof}

Let $\mathcal{B}=(\{1,2,3\}, c,\gamma_1, \gamma_2, \gamma_3, \Phi_1, \Phi_2, \Phi_3)$ be an EBG such that $\Phi_1 = \{p\}$, $\Phi_2=\{q\}$, $\Phi_3=\{r\}$, $\gamma_1= p, \gamma_2 = q, \gamma_3 = r$, $c_1 (v_{ \neg p\wedge \neg q \wedge \neg r}) =10 $, $c_1 (v^{\prime}) =0$ for $v^{\prime} \neq v_{\neg p\wedge \neg q \wedge \neg r}$, $c_1(x) =c_2(x) =c_3(x)$ for each $x\in V$. The outcome $v_{p \wedge q \wedge r}$ is clearly a unique pure strategy Nash equilibrium. Suppose now that it is also a surviving equilibrium and that the transfer function $\beta$ is part of the solution. But then there exists some player $i$ for which $u^{\beta(\mathcal{B})}_i(v_{p \wedge q \wedge r}) < 30+\epsilon$. This means that $i$ is better off deviating to a transfer function $(\beta^{\prime}_i,\beta_{-i})$ such that $\beta^{\prime}_i(v_{\neg p \wedge \neg q \wedge \neg r})+ c_i(v_{\neg p \wedge \neg q \wedge \neg r}) = 30$. No matter what equilibrium will be played in the resulting subgame, $i$ will be boosting his payoff at $v_{p\wedge q \wedge r}$ to $30 + \epsilon$. Contradiction. \end{proof}

So not only have we constructed a non-surviving pure strategy Nash equilibrium in the case of more than two players, but also an outcome that is shared joint goal among all players, which was straightforward truth for boost factor profiles $\omega \in \Omega\uparrow$ (Proposition \ref{prop:joint}). As anticipated, with boost factor profiles in $\Omega \downarrow$ players can  increase their cost at non-goal states for the sole reason of increasing their payoff at goal states, which is exactly what happens here. Figure \ref{non-survival:boolean} displays once more this fact.

\begin{figure}[htb]\hspace*{\fill}%
\begin{game}{2}{2}[\yellow{$1$}]
     \> \blue{$L$}   \> $R$\\
\red{$U$} \> { $0,0,0$   }\>{$0,0,0$}\\
$D$   \>{${0,0,0}$}  \>{$0,0,0$}
\end{game}\hspace*{\fill}%
\begin{game}{2}{2}[$2$]
     \> \blue{$L$}   \> $R$\\
\red{$U$} \> { $0,0,0$   }\>$0,0,0$\\
$D$   \>$0,0,0$  \>{$10,10,10$}
\end{game}\hspace*{\fill}%
\caption{A three player game, with the third player (\textcolor{yellow}{Table}) choosing the matrix to be played. \textcolor{red}{Row} wants any outcome consistent with $U$ to be realized, \textcolor{cyan}{Column} any outcome consistent with $L$ and \textcolor{yellow}{Table} any outcome consistent with $1$.  Notice that all players are in control of their own goal satisfaction.  To ease readability we avoid displaying all coalitional colours which are, as expected, \textcolor{purple}{$\{Row, Column\}$},  \textcolor{green}{$\{Table, Column\}$},  \textcolor{brown}{$\{Row, Column, Table\}$},  \textcolor{orange}{$\{Row, Table\}$}. Rather, we only label the strategies corresponding to individual players' goals. The brown outcome $(U,L,1)$ is \textcolor{red}{Row}'s, \textcolor{cyan}{Column}'s and \textcolor{yellow}{Table}'s joint goal and happens to be a dominant strategy equilibrium of the game. However it is not a surviving equilibrium as, no matter what the distribution of payoffs after the pre-play phase looks like, there will always be a player that can increase his cost at outcome $(D,R,2)$ for the sole reason to increase his payoff at outcome $(U,L,1)$.}
\label{non-survival:boolean}
\end{figure}

Sufficient conditions for survival in the case of boolean games are slightly more demanding.

Let us call an outcome $v$ {\bf shareable} if we can assign to each $i$ with $v \models \gamma_i$ a unique outcome ${\bf v}^{i}$ such that:

\begin{itemize}

\item for all $j\in N$, $(v_{j},{\bf v}^{i}_{-j})\neq v$. In other words, we focus on outcomes that cannot be reached from $v$ by an individual deviation.

\item for all $v^\prime\in V$, $\sum_{j\in N} c_i(v^{\prime}) \leq \sum_{j\in N} c_j({\bf v}^{i})$. In other words, each player realizing a goal in $v$ is identified with a unique outcome  where the aggregated cost is maximal.

\end{itemize}

We call it moreover {\bf potentially shareable} if there exists a cost function $c^{*}$ such that the outcome $v$ of $(N, \Phi, c^{*}, \{\gamma_i\}_{i\in N}, \{\Phi_i\}_{i\in N})$ is shareable.
EBGs with shareable outcomes display the property we are after.

\begin{proposition}[Survival]\label{prop:survival2}

Let $(\mathcal{B},\{B_i\}_{i\in N})$ be an EBG. A pure strategy Nash equilibrium $v$ of $\mathcal{B}$ survives whenever $v$ is shareable and $u_i(v) \geq \hat{s}_i$ for each $i\in N$.

\end{proposition}

\begin{proof}

We need to construct a subgame perfect equilibrium of the two-phase game where $v$ is played. On the equilibrium path let players play the profile $(\beta^{\prime}, v)$ where: [1] 
for all $j \not\in \{i\in N \mid v \models \gamma_i\}, v^{\prime} \in V$, we have that, for all $k\neq j$, $\beta^{\prime}_j(v^{\prime},k)=\beta^{0}_{j}(v^{\prime},k)$; [2] for all $j \in \{i\in N \mid v \models \gamma_i\}$ we have that $\beta^{\prime}_j({\bf v}^j,k)= c_k({\bf v}^j) $ and $\beta^{\prime}_j(v^{\prime},k) = \beta_j^0(v^{\prime},k)$, for all $v^{\prime} \neq {\bf v}^{j}$, $k\neq j$. Off the equilibrium path, if a single player deviates from $\beta^{\prime}$ pick the worst Nash equilibrium for that player and have that played in the continuation. If more than a player deviates from $\beta^{\prime}$, play any Nash equilibrium. We can observe that no player $i$ can get more than $\hat{s}_i$ by deviating from $\beta^{\prime}$, given the choices played in each subgame. Moreover $v$ is a pure strategy Nash equilibrium of $\beta^{\prime}(\mathcal{B})$. \end{proof}

To see that there are EBGs where pure strategy Nash equilibria do survive, notice that a large enough game, where all outcomes have zero cost for all players and where there is a common goal, has a shareable outcome and players get at least their solo payoff in the corresponding strategy profile. \footnote{The size requirement is not particularly demanding. Every boolean game with $|\Phi_i| \geq 2$ for each $i$ has a potentially shareable outcome. If all the costs are uniform, that outcome is also shareable.}  

\section{An integrated framework}\label{sec:integration}


If we take the point of view of an external authority that would like a certain outcome of an EBG to be rationally chosen by players, the question is not only whether that outcome can be turned into a Nash equilibrium, but also whether that Nash equilibrium is bound to survive.
The purpose of this section is to generate taxation mechanisms that guarantee equilibria to survive when players play rational transfers to one another and play rationally in the game that is updated both with the transfers made and with the taxation mechanism.

The procedure we present  takes an outcome that is consistent with players' goals and turns it into a surviving Nash equilibrium by appropriately employing taxation mechanisms. 

\begin{algorithm}\textcolor{white}{steps}\label{alg:algo}
\begin{description}

\item[Input] An outcome $\sigma$ of an instantiated strategic game with goals $(\mathcal{S},\{G_i\}_{i\in N})(\omega)$ with $\mathcal{S} = (N, \{\Sigma_i\}_{i\in N}, \pi)$ and such that for each $i\in N$ either $\sigma \in G_i$ or for no $\sigma^{\prime}_i$ we have that $(\sigma_{-i}, \sigma^{\prime}_i) \in G_i$.
\item[Output] A taxation mechanism $\alpha$ on $\mathcal{S}$.

\item[Steps]  \textcolor{white}{ciao}
\begin{enumerate}

\item {\bf Let} $\alpha_i(\sigma^{\prime}) = 0$,  for each $i\in N, \sigma^{\prime}\in \Sigma$;
\item {\bf While} for some $i\in N$ we have that \\ $u^{(\alpha(\mathcal{S}),\{G_i\}_{i\in N})(\omega)}_i(\sigma) < \hat{s}^{((\alpha(\mathcal{S}), \{G_i\}_{i\in N}),\{T_i\}_{i\in N})(\omega)}_i$,\\  {\bf do} $\alpha_i(\sigma^{\prime}):=\alpha_i(\sigma^{\prime})+1$, for each $\sigma^{\prime}\neq \sigma \in \Sigma$;  
\item {\bf Return} $\alpha$.

\end{enumerate}

\end{description}

\end{algorithm}

\begin{proposition}[Survival by taxation]

 Let  $\sigma$ be an outcome of an instantiated strategic game with goals $(\mathcal{S},\{G_i\}_{i\in N})(\omega)$ with $\mathcal{S} = (N, \{\Sigma_i\}_{i\in N}, \pi)$ and such that for each $i\in N$ either $\sigma \in G_i$ or for no $\sigma^{\prime}_i$ we have that $(\sigma_{-i}, \sigma^{\prime}_i) \in G_i$.
 There exists a taxation mechanism $\alpha$ such that $\sigma$ is a surviving equilibrium of  $((\alpha(\mathcal{S}), \{G_i\}_{i\in N}),\{T_i\}_{i\in N})(\omega)$.
\end{proposition}

\begin{proof}

To see that Algorithm \ref{alg:algo} guarantees this fact we only need to observe that the construction of $\alpha$ at step 1 and step 2 ensures that the payoff $u^{(\alpha(\mathcal{S}), \{G_i\}_{i \in N})(\omega)}_i(\sigma)$ will eventually reach $\hat{s}^{((\alpha(\mathcal{S}), \{G_i\}_{i \in N}), \{T_i\}_{i \in N})(\omega)}_i$. \end{proof}

For the case of boolean games we need slightly more reasoning effort, showing that the property of shareability is preserved.

\begin{proposition}[Survival by taxation (Boolean Games)]

Let $v$ be potentially shareable outcome of a boolean game $\mathcal{B}=(N, \Phi, c, \{\gamma_i\}_{i\in N}, \{\Phi_i\}_{i\in N})$ such that for each i either $v \models \gamma_i$ or for no $v^{\prime}_i$ we have that $(v_{-i}, v^{\prime}_i) \models \gamma_i$. There exists a taxation mechanism $\alpha$ such that $v$ is a surviving equilibrium of the EBG $(\alpha(\mathcal{B}),\{B_i\}_{i\in N})$.
\end{proposition}

\begin{proof}
To see that Algorithm \ref{alg:algo} guarantees this fact we only need to observe that the construction of $\alpha$ at step 1 ensures that the outcome $v$ of game $\alpha(\mathcal{B})$ is shareable, and the update of $\alpha$ at step 2 ensures that the payoff $u^{\alpha(\mathcal{B})}_i(v)$ will eventually reach $\hat{s}^{(\alpha(\mathcal{B}),\{B_i\}_{i\in N})}_i$ while keeping $v$ shareable. \end{proof}
\section{Variants: lexicographic preferences over mixed strategies}\label{sec:lexico}

As already argued in the introductory section, the definition of expected utility is ill-formed when utilities range over extended reals, i.e., the set $\mathbb{R} \cup \infty$, due to the existence of strategy profiles associated with infinite utility.

In this section we present a fairly reasonable notion of preference over mixed strategies that takes the infinitary nature of goal states into account, and argue, consistently with many results already available in the literature, that to analyze this type of preferences we need to go beyond the realm of normal form games.

\begin{figure}[htb]\hspace*{\fill}%
\begin{game}{2}{2}
     \> $L$   \> $R$\\
$U$ \> {\red{$0,0$}}\>\blue{{$0,0$}}\\
$D$   \> \blue{$1,0$}   \>\red{$0,0$}
\end{game}\hspace*{\fill}%
\caption{Lexicographic preference over mixed strategies and no mixed strategy Nash equilibrium}
\label{noNash}
\end{figure}

\begin{definition}[Lexicographic prefence over mixed strategies]\label{lexico}

Let $(\mathcal{S}, \{G_i\})$ be a strategic game with goals and let $G_i(\delta)$ the probability of mixed strategy profile $\delta$ assigned to profiles in $G_i$. The lexicographic preference $\leq^{LEX_{\Delta}}$ over mixed strategies, with strict and reverse counterparts denoted as usual, is defined as follows, for any two mixed strategy profiles $\delta, \delta^{\prime} \in \Delta$ available at $\mathcal{S}$:

$$\delta \leq^{LEX_{\Delta}}_i \delta^{\prime} \mbox { if and only if } (G_i(\delta) < G_i (\delta^{\prime}) \mbox{ OR } ( G_i(\delta) = G_i (\delta^{\prime}) \mbox{ AND }  E_i(\delta) \leq E_i(\delta^{\prime}))) $$

\end{definition}

In words a mixed strategy profile $\delta$ is better for $i$ than a mixed strategy profile $\delta^{\prime}$ if:

\begin{itemize} 

\item either the probability of reaching a $i$'s goal state in $\delta$ is higher than that of $\delta^{\prime}$;

\item the probability of reaching $i$'s  goal state is the same in $\delta$ and $\delta^{\prime}$ but the expected utility for $i$ in $\delta$ is higher than the one in $\delta^{\prime}$.

\end{itemize}

Notice that the newly defined preference is actually a total preorder.

\begin{proposition}[Properties]

$\leq^{LEX_{\Delta}}$ is transitive and complete.

\end{proposition}

\begin{proof}
Straightforward consequence of the definition.

\end{proof}

Nevertheless it cannot be represented by a preference relation in a strategic game (in the sense of \cite{Rubinstein}), as the following result shows.

\begin{proposition}

There exists a strategic game with goals $(\mathcal{S}, \{G_i\}_{i\in N})$ with induced lexicographic preference relation over mixed strategies $\leq^{LEX_{\Delta}}$ that is not representable by a strategic game.

\end{proposition}

\begin{proof}

Consider the strategic game with goals displayed in Figure  \ref{noNash} and consider the relation $ \leq^{LEX_{\Delta}}_i$ constructed following Definition \ref{lexico}. 
Assume, for the sake of contradiction, that there is a Nash equilibrium profile $\delta^{*}$. Consider now the following cases

\begin{itemize}

\item $\delta^{*}(R) > \delta^{*}(L)$, i.e., the probability for Column to play $R$ is higher than to play $L$. But then 
we must have that $\delta^{*}(D) > \delta^{*}(U)$ as $(R,D)>^{LEX_{\Delta}}_{Row} (R,U)$. But if this is the case then 
$\delta^{*}(L) > \delta^{*}(R)$ as $(L,D)>^{LEX_{\Delta}}_{Column} (R,D)$. Contradiction.

\item $\delta^{*}(L) > \delta^{*}(R)$, similar reasoning. Contradiction.

\item $\delta^{*}(L) = \delta^{*}(R)$. But then 
we must have that $\delta^{*}(D) > \delta^{*}(U)$ as $((0.5 L; 0.5 R),U)<^{LEX_{\Delta}}_{Row} ((0.5 L; 0.5 R),D)$. But then again $\delta^{*}$ cannot be an equilibrium unless $\delta^{*}(L) < \delta^{*}(R)$. Contradiction.

\end{itemize}

In conclusion, $\delta^{*}$ is not a Nash equilibrium. Contradiction. Now the fact that $(\mathcal{S}, \{G_i\}_{i\in N})$ with induced lexicographic preference relation over mixed strategies $\leq^{LEX_{\Delta}}_i$ is not representable by a strategic game follows from Nash's theorem \cite{nash}.
\end{proof}

The proposition shows that the adoption of the extremely natural preference relation $ \leq^{LEX_{\Delta}}_i$ induces structures that are not amenable to standard equilibrium analysis.

We consider this a double-edged result. On the one hand, $ \leq^{LEX_{\Delta}}_i$ seems to us an extremely natural preference relation that is worth considering, even at the cost of sacrificing standard game-theoretic analysis. 
On the other hand, the fact that $ \leq^{LEX_{\Delta}}_i$ is not in general analyzable within the framework of normal form games does not enable a connection and a comparison with relevant results in the field, such as Jackson and Wilkie's endogenous games, which rely on the existence of Nash equilibria in their two-phase game, and a whole host of theoretical effort devoted to the study of pre-play negotiations and similar setups \cite{EP11,kal1,gut1, gut2,var1,var2, farrell,maskin}.

The following results show that  $ \leq^{LEX_{\Delta}}_i$ is otherwise rather well-behaved.

\begin{proposition}

Let $IESDS(\mathcal{S}, \{G_i\}_{i\in N})$ be the procedure of iterated elimination of strictly dominated strategies applied to $(\mathcal{S}, \{G_i\}_{i\in N})$ using the preference relation $ \leq^{LEX_{\Delta}}_i$ and $IEWDS(\mathcal{S}, \{G_i\}_{i\in N})$ its weak variant.

\begin{itemize}

\item If $(\mathcal{S}^{\prime}, \{G^{\prime}_i\}_{i\in N})$ is an outcome of $IESDS(\mathcal{S}, \{G_i\}_{i\in N})$, then $\delta$ is Nash equilibrium of $(\mathcal{S}, \{G_i\}_{i\in N})$ if and only if it is a Nash equilibrium of 
$(\mathcal{S}^{\prime}, \{G^{\prime}_i\}_{i\in N})$;

\item If $(\mathcal{S}, \{G_i\}_{i\in N})$ is solved by $IESDS$ then the resulting joint strategy is a unique Nash equilibrium of $(\mathcal{S}, \{G_i\}_{i\in N})$;

\item If $(\mathcal{S}^{\prime}, \{G^{\prime}_i\}_{i\in N})$ is an outcome of $IEWDS(\mathcal{S}, \{G_i\}_{i\in N})$, and $\delta$ is Nash equilibrium of $(\mathcal{S}^{\prime}, \{G^{\prime}_i\}_{i\in N})$ then it is a Nash equilibrium of 
$(\mathcal{S}, \{G_i\}_{i\in N})$;

\item If $(\mathcal{S}, \{G_i\}_{i\in N})$ is solved by $IESDS$ then the resulting joint strategy is a Nash equilibrium of $(\mathcal{S}, \{G_i\}_{i\in N})$.
\end{itemize}

\end{proposition}

\begin{proof}

The proof follows the standard procedure. The reader is referred to  \cite{apt} for more detail.

\end{proof}

\section{Conclusion}\label{sec:conclusion}

We have studied strategic games where players are endowed with designated goal states and with the possibility of offering side-payments to their fellow players in order to influence their decision-making. The perspective we have taken integrates the framework of strategic games with goals, a generalization of the boolean games studied in artificial intelligence, with that of endogenous games with side-payment,s studied in game theory. We have seen that the resulting games display specific properties that make them worth studying in their own sake (Propositions \ref{prop:static1}, \ref{prop:static2}, \ref{prop:static3} and \ref{prop:dynamic}) and the classical results available on Nash equilibria survival do not generalize (Proposition \ref{prop:nonsurvival}). We have however provided sufficient conditions that Nash equilibria need to have in order to survive (Proposition \ref{prop:survival2}), independently of the number of players involved. We have also shown that, with an appropriate use of taxation mechanisms, every outcome consistent with players' goals can be turned into a surviving Nash equilibrium (Algorithm \ref{alg:algo}).
Future research efforts will be devoted to studying the interaction between side-payments in strategic games with goals and more realistic taxation mechanism that carry out imperfect redistribution of wealth, i.e., extract payoff units to some players at certain outcomes redistributing a part of it to possibly different players at possibly different outcomes. Attention will also be paid to the relation with mechanism design and the algorithmic properties of the procedures under study.

\section*{Acknowledgments}

Paolo Turrini acknowledges the support of the IEF Marie Curie fellowship "Norms in Action: Designing and Comparing Regulatory Mechanisms for Multi-Agent Systems" (FP7-PEOPLE-2012-IEF, 327424 "NINA"). 

He is indebted to  Ulle Endriss, Valentin Goranko,  Umberto Grandi, Davide Grossi, Paul Harrenstein, Wojtek Jamroga, J\'{e}r\^{o}me Lang and Michael Wooldridge, for their feedback to the paper "Endogenous Boolean Games", presented at the 23rd International Joint Conference on Artificial Intelligence (IJCAI 2013). 

\newpage
\bibliographystyle{alpha}
\bibliography{EBG-Bibliography}

\end{document}